\documentclass[12pt]{amsart}

\usepackage{amsmath,amsthm,amssymb}
\usepackage{verbatim}
\usepackage{tikz}
\usepackage{float}
\usepackage{mathtools}
\usepackage{subcaption}
\usepackage{algorithm}
\usepackage{algorithmic}
\usepackage{MnSymbol}
\usepackage{graphicx}
\usepackage[left=3cm,top=3cm,right=3cm,bottom=3cm]{geometry}

\newtheorem{theorem}{Theorem}
\newtheorem{lemma}[theorem]{Lemma}
\newtheorem{corollary}[theorem]{Corollary}

\theoremstyle{definition}

\newcommand{\qedclaim}{\hfill $\diamond$ \medskip}

\begin{document}

\title{The Game of Cops and Eternal Robbers}\thanks{The first and second authors were supported by NSERC. The research of the fourth author was supported by ANR projects GAG and DISTANCIA (ANR-17-CE40-0015 and ANR-14-CE25-0006, respectively).}

\author[A.\ Bonato]{Anthony Bonato}
\author[M.\ Huggan]{Melissa Huggan}
\author[T.\ Marbach]{Trent Marbach}
\address[A1, A2, A3]{Ryerson University, Toronto, Canada}
\email[A1]{(A1) abonato@ryerson.ca}
\email[A2]{(A2) melissa.huggan@ryerson.ca}
\email[A3]{(A3) trent.marbach@ryerson.ca}
\author[F.\ Mc~Inerney]{Fionn Mc~Inerney}
\address[A4]{Laboratoire d'Informatique et Syst\`emes, Aix-Marseille Universit\'e, CNRS, and Universit\'e de Toulon Facult\'e des Sciences de Luminy, Marseille, France}
\email[A4]{(A4) fmcinern@gmail.com}

\begin{abstract}
We introduce the game of Cops and Eternal Robbers played on graphs, where there are infinitely many robbers that appear sequentially over distinct plays of the game. A positive integer $t$ is fixed, and the cops are required to capture the robber in at most $t$ time-steps in each play. The associated optimization parameter is the eternal cop number, denoted by $c_t^{\infty},$ which equals the eternal domination number in the case $t=1,$ and the cop number for sufficiently large $t.$ We study the complexity of Cops and Eternal Robbers, and show that game is {\bf NP}-hard when $t$ is a fixed constant and {\bf EXPTIME}-complete for large values of $t$. We determine precise values of $c_t^{\infty}$ for paths and cycles. The eternal cop number is studied for retracts, and this approach is applied to give bounds for trees, as well as for strong and Cartesian grids.
\end{abstract}

\maketitle

\section{Introduction}

Graph searching is a rapidly growing area of study within graph theory, with Cops and Robbers and its variants among the most actively studied vertex-pursuit games. The cop number is the main graph parameter in this area, and its consideration has led to a number of deep open problems, such as Meyniel's conjecture on the largest asymptotic order of the cop number of a connected graph. See the books \cite{BN,B} for background on Cops and Robbers and graph searching, and the surveys \cite{BB,BM} for background on Meyniel's conjecture and topological directions, respectively.

In the game of Cops and Robbers, if the robber is caught, then the game finishes with a win for the cops. However, we may consider a game where there are an \emph{infinite} number of robbers appearing, which must be caught sequentially over plays of the game. An analogous game with such infinite attacks is the {\it Eternal Domination game}~\cite{GoddardHH05}. In this game played on a graph $G$, there is a set of guards and an attacker. The guards first occupy a set of vertices, and the players alternate turns, starting with the attacker choosing a vertex $v$ to attack. After an attack, each of the guards may move to an adjacent vertex or remain on their current vertex. The guards win if at the end of each of their turns, at least one guard occupies the attacked vertex $v$; otherwise, the attacker wins. The {\it eternal domination number} of a graph $G$, denoted by $\gamma^{\infty}_{\mathrm{all}}(G)$, is the minimum number of guards needed to win. Note that this is the variant of the Eternal Domination game where multiple guards may occupy the same vertex and every guard may move on the guards' turn. Given a graph $G$ and $k$ a positive integer, deciding whether $\gamma^{\infty}_{\mathrm{all}}(G)\leq k$ is {\bf NP}-hard~\cite{BDEMY17} but not known to be in {\bf NP}. For paths and cycles on $n$ vertices, denoted by $P_n$ and $C_n$, respectively, $\gamma^{\infty}_{\mathrm{all}}(P_n)=\left\lceil \frac{n}{2} \right\rceil$ and $\gamma^{\infty}_{\mathrm{all}}(C_n)=\left\lceil \frac{n}{3} \right\rceil$; see~\cite{GoddardHH05}. There is a linear-time algorithm for computing the eternal domination number of trees~\cite{KM09}. There have been several papers on eternal domination in grids; see, for example, \cite{BeatonFM13,Messinger17,LMS19,MNP19}.

We consider in the present work a new game called {\it Cops and Eternal Robbers}, which is both a new variant of the game of Cops and Robbers and a generalization of the Eternal Domination game. In this game, there are two players: a set of cops and a set of robbers. Movement of the cops and robbers is as in the classical game, where the players are on vertices and move vertex-to-vertex along edges (and may pass by remaining at their current vertex). The cops move first, and then moves alternate between the players. The cops \emph{capture} the robber if they land on his vertex. A {\it play} consists of one game of Cops and Robbers and finishes if the robber is captured. If a robber is captured, then the cops begin the next play on the vertices they occupied at the end of the previous play. In each play, a new robber appears on any unoccupied vertex. The cops \emph{win} if they capture the robber in every play; otherwise, the robber wins (that is, the robber wins if they avoid capture in any one play).

If the time taken to capture the robber is not a concern, then Cops and Eternal Robbers is equivalent to playing Cops and Robbers sequentially over an infinite set of plays, and the number of cops required to capture the robber will be consistent with the classical game. Note that whatever the optimal configuration to capture, the cops could regroup to that configuration on subsequent plays, then capture the robber in finitely many steps. Hence, our restriction on the cops is to capture the robber within a prescribed time in each play. A \emph{time-step} consists of both the turn of the cops and the turn of the robber. We allow a fixed amount of time, say $t$, for the cops to capture the robber on a graph $G$. For $t$ a positive integer, let $c_{t}(G)$ be the minimum number of cops required to capture the robber in at most $t$ time-steps for a single play. Define $c_t^{\infty}(G)$ to be the minimum number of cops needed to ensure capturing the robber in at most $t$ time-steps in each of the infinitely many plays. We refer to this graph parameter as the \emph{eternal cop number} (inspired by the eternal domination number).
Observe that $c_t^{\infty}(G)\geq c(G),$ and that $c_t^{\infty}(G)$ is well-defined since $c_t^{\infty}(G)\leq n$, where $n$ is the order of $G$. Interestingly, when $t=1$, the game of Cops and Eternal Robbers is equivalent to the Eternal Domination game. In particular, for a graph $G$, we have that $c_1^{\infty}(G)=\gamma_{\mathrm{all}}^{\infty}(G)$.

In~\cite{BonatoGHK09}, they defined $\mathrm{capt}_k(G)$ (simplified to $\mathrm{capt}(G)$ when $k=c(G)$) to be the minimum number of time-steps it takes for $k\geq c(G)$ cops to capture the robber in $G$. Observe that for all integers $k\geq c(G)$ we have that
\begin{equation}
c_{2\mathrm{capt}_k(G)}^{\infty}(G)\leq k \quad \text{and} \quad c_{2\mathrm{capt}(G)}^{\infty}(G) \leq c(G). \label{xx}
\end{equation}

Note that in Cops and Eternal Robbers, the plays are not independent: the outcome of the previous play determines the placement of the cops for the next play. The setting is similar to the concept of \emph{conjoined rulesets} for combining a finite number of combinatorial games played sequentially~\cite{HN19}. Given two rulesets $G$ and $H$, the conjoined ruleset of $G$ and $H$ is that players play the first phase of the game under ruleset $G$ and when play is no longer possible using ruleset $G$, they switch to the second phase and play under the ruleset of $H$. In other words, players play a finite number of games, sequentially, on the terminal board of the previous game. In the game of Cops and Eternal Robbers, the board state changes by the robber showing up again and again, as well as the possibility of infinite play. From conjoined games, we know that players may deliberately lose certain games to be the first to play in the next game (knowing they will win the next game). The idea parallels the robber player deliberately losing the first $k$ plays, so that they can win on play $k+1$ (sacrificing themselves for the greater good of them and their allies; see the proof of Theorem~\ref{thm:np-hard} in Section~\ref{sec:complexity}).

We now outline the results of this paper. In Section~\ref{sec:complexity}, we study the complexity of Cops and Eternal Robbers. In particular, we show that the game is {\bf NP}-hard when $t$ is a fixed constant and {\bf EXPTIME}-complete for large values of $t$. In Section~\ref{sec:paths&cycles}, we determine $c_t^{\infty}$ precisely for paths and cycles. In Section~\ref{sec:retracts}, we study $c_t^{\infty}$ using retracts and apply our techniques to bounding the eternal cop number on trees, and strong and Cartesian grids. In the final section, we conclude with further work and open problems.

Throughout this paper, all graphs are undirected, finite, with no multiple edges, and reflexive (the latter property allows players to pass). All logarithms considered throughout are to the base 2 unless otherwise stated. The symbol $\mathbb{N}^{*}$ refers to the set of positive integers. For a general reference on graph theory, the reader is directed to \cite{West}.

\section{Complexity of computing the eternal cop number}\label{sec:complexity}

Recall that {\bf EXPTIME} is the class of decision problems solvable in exponential time. Given a graph $G$ and $k$ a positive integer, deciding whether $c(G)\leq k$ is {\bf EXPTIME}-complete~\cite{Kinnersley15}. We have an immediate, analogous result for Cops and Eternal Robbers for some instances of $t$; namely, those where $t$ is at least as large as the total number of possible configurations for all of the agents.

\begin{theorem}\label{thm:conf}
Let $G$ be a graph on $n$ vertices and $k$ a positive integer. We then have $c_{n\binom{n+k-1}{k}}^{\infty}(G)\leq k$ if and only if $c(G)\leq k$.
\end{theorem}

\begin{proof}
In any optimal strategy for the cops, no configuration of both the cops and the robber is repeated in the same play, as otherwise, all of the moves in between the two occurrences of the same configuration were pointless and thus, not optimal. The proof then follows since there are a total of $n\binom{n+k-1}{k}$ possible configurations of both the cops and the robber.
\end{proof}

Theorem~\ref{thm:conf} implies the following complexity result for Cops and Eternal Robbers.

\begin{corollary}
Given a graph $G$ on $n$ vertices and $k$ a positive integer (not fixed), the problem of deciding whether $c_{n\binom{n+k-1}{k}}^{\infty}(G)\leq k$ is {\bf EXPTIME}-complete.
\end{corollary}

\begin{proof}
Since there are a total of $n\binom{n+k-1}{k}=O(n^{k+1})$ possible configurations of both the cops and the robber, the problem is in {\bf EXPTIME}. The {\bf EXPTIME}-hardness of the problem comes from Theorem~\ref{thm:conf} and the fact that given a graph $G$ and a positive integer $k$, deciding whether $c(G)\leq k$ is {\bf EXPTIME}-hard~\cite{Kinnersley15}.
\end{proof}

We now turn to the case where $t$ is a fixed constant. The following lemma will be used in the proof of the next theorem. To simplify the lemma and its proof, we first define the following. A {\it sum-decreasing sequence} has the property that for any term $n$ in the sequence, the term $n$ is a positive integer and the sum of all of the terms after $n$ in the sequence is strictly less than $n$. Let $\mathrm{maxseq}(t)$ be the maximum length of a sum-decreasing sequence with $t$ as its first term.

\begin{lemma}\label{lem:sequence}
Given $t$ a positive integer, we have $\mathrm{maxseq}(t)=\left\lfloor\log{t}\right\rfloor+1$.
\end{lemma}

\begin{proof}
The proof is done by induction on $t$. The result is immediate for $t=1$. Suppose, for all $k\geq 0$, that $\mathrm{maxseq}(k)=\left\lfloor\log{k}\right\rfloor+1$. We will make use of the following claim (the maximum length of a sum-decreasing sequence is monotonically decreasing in terms of its first term $t$) to prove this lemma.

\smallskip

\noindent \emph{Claim:}
For positive integers $t_0,t_1$ such that $t_1>t_0$, $\mathrm{maxseq}(t_1)\geq \mathrm{maxseq}(t_0)$.

\smallskip

\noindent For the proof of the claim and for the sake of contradiction, assume $t_1>t_0$ and $\mathrm{maxseq}(t_0) > \mathrm{maxseq}(t_1)$. Taking the sequence with $t_0$ as its first term and replacing $t_0$ by $t_1$ gives a sum-decreasing sequence (of the same length) with $t_1$ as its first term. The proof of the claim follows from this contradiction.

\smallskip

We first show that $\mathrm{maxseq}(k+1)\geq \lfloor \log{(k+1)} \rfloor +1$. By the Claim and the inductive hypothesis, $\mathrm{maxseq}(k+1)\geq \lfloor \log{k} \rfloor +1$. If $\lfloor \log{(k+1)} \rfloor = \lfloor \log{k} \rfloor$, then we are done. Otherwise, $\lfloor \log{(k+1)} \rfloor \neq \lfloor \log{k} \rfloor$, and thus, $\log{(k+1)} \in \mathbb{N^*}$ and the sequence $k+1,\frac{k+1}{2},\frac{k+1}{4},\ldots,1$ is a sum-decreasing sequence of length $\lfloor \log{(k+1)} \rfloor+1$.

We next show that $\mathrm{maxseq}(k+1)\leq \lfloor \log{(k+1)} \rfloor +1$. Assume, for purpose of contradiction, that $\mathrm{maxseq}(k+1)> \lfloor \log{(k+1)} \rfloor +1$. By the inductive hypothesis, $\mathrm{maxseq}(k+1)\leq 1+ \lfloor \log{k} \rfloor +1=\lfloor \log{k} \rfloor +2$. If $k+1$ is a power of $2$, then $\lfloor \log{(k+1)} \rfloor=\lfloor \log{k} \rfloor +1$, in which case we are done. Otherwise, $k+1$ is not a power of 2, and so $\lfloor \log{(k+1)} \rfloor = \lfloor \log{k} \rfloor$.

Let $\ell$ be the smallest integer such that $\lfloor \log{(k+1)} \rfloor = \log{\ell}$. By the inductive hypothesis, any sum-decreasing sequence, with $k+1$ as its first term and having length strictly greater than $\lfloor \log{(k+1)} \rfloor +1=\log{\ell}+1$, must have an integer $x\geq \ell$ as a second term. We then have that, $x>\frac{k+1}{2}$ since $k+1<2\ell$. Let $y$ be the third term. We find that, $y<\frac{k+1}{2}$ since $x>\frac{k+1}{2}$. By the inductive hypothesis and the assumption that $\mathrm{maxseq}(k+1)> \lfloor \log{(k+1)} \rfloor +1$, we have that $y\geq \frac{\ell}{2}$ and so, $y>\frac{k+1}{4}$. The fourth term must be less than $\frac{k+1}{4}$ since $x+y > \frac{3(k+1)}{4}$. Continuing in this fashion, for all $i>2$, the $i^{th}$ term must then be less than $\frac{k+1}{2^{i-2}}$. Therefore, the $(\lfloor \log{(k+1)} \rfloor +2)^{th}$ term must be less than $\frac{k+1}{2^{\log{(k+1)}+2-2}}=1$, a contradiction.
\end{proof}

Now we can prove the main result of this section.

\begin{theorem}\label{thm:np-hard}
Given a graph $G$, $k\in \mathbb{N^*}$, and a fixed constant $t\in \mathbb{N^*}$, the problem of deciding whether $c_t^{\infty}(G)\leq k$ is {\bf NP}-hard.
\end{theorem}

\begin{proof}
The reduction we use is from the well-known {\bf NP}-hard set cover problem~\cite{GJ79}. The set cover problem is stated as follows: given an integer $k>0$, a set of elements $B=\{b_1,\ldots,b_{\alpha}\}$ called the \emph{universe}, and subsets $S_1,\ldots,S_{\beta}$ of the universe whose union forms the universe, does there exist a collection of at most $k$ subsets from $S_1,\ldots,S_{\beta}$ such that their union forms the universe? Such a set would be called a {\it set cover} of size at most $k$. We construct, in polynomial time, a graph $G$ where there exists a set cover of size at most $k$ if and only if $c_t^{\infty}(G)\leq k+\left\lfloor \log{t} \right\rfloor+1$.

To construct $G$, take the typical graph associated to set cover; that is, there is a vertex for each subset $S_1,\ldots,S_{\beta}$ and for each element $b_1,\ldots, b_{\alpha}$. A vertex associated with an element is adjacent to a vertex associated with a subset if and only if that element is contained in the subset. From each element vertex attach a path of size $t-1$ (the distance from an element vertex to the other end of the path attached to it is then $t-1$) which will be called an \emph{attached} path. In addition, add $\lfloor \log{t} \rfloor + 1$ paths of size $t$ with exactly one end of each path being considered as an element vertex. The path of size $t-1$, obtained by removing the element vertex of any of these paths, will be referred to as an attached path. For each of these additional $\lfloor \log{t} \rfloor + 1$ paths of size $t$, make its element vertex adjacent to all of the other subset vertices. Finally, form a clique from all of the subset vertices. See Figure~\ref{fig:np-hard}.
\begin{figure}
\centering
\includegraphics[width=\textwidth]{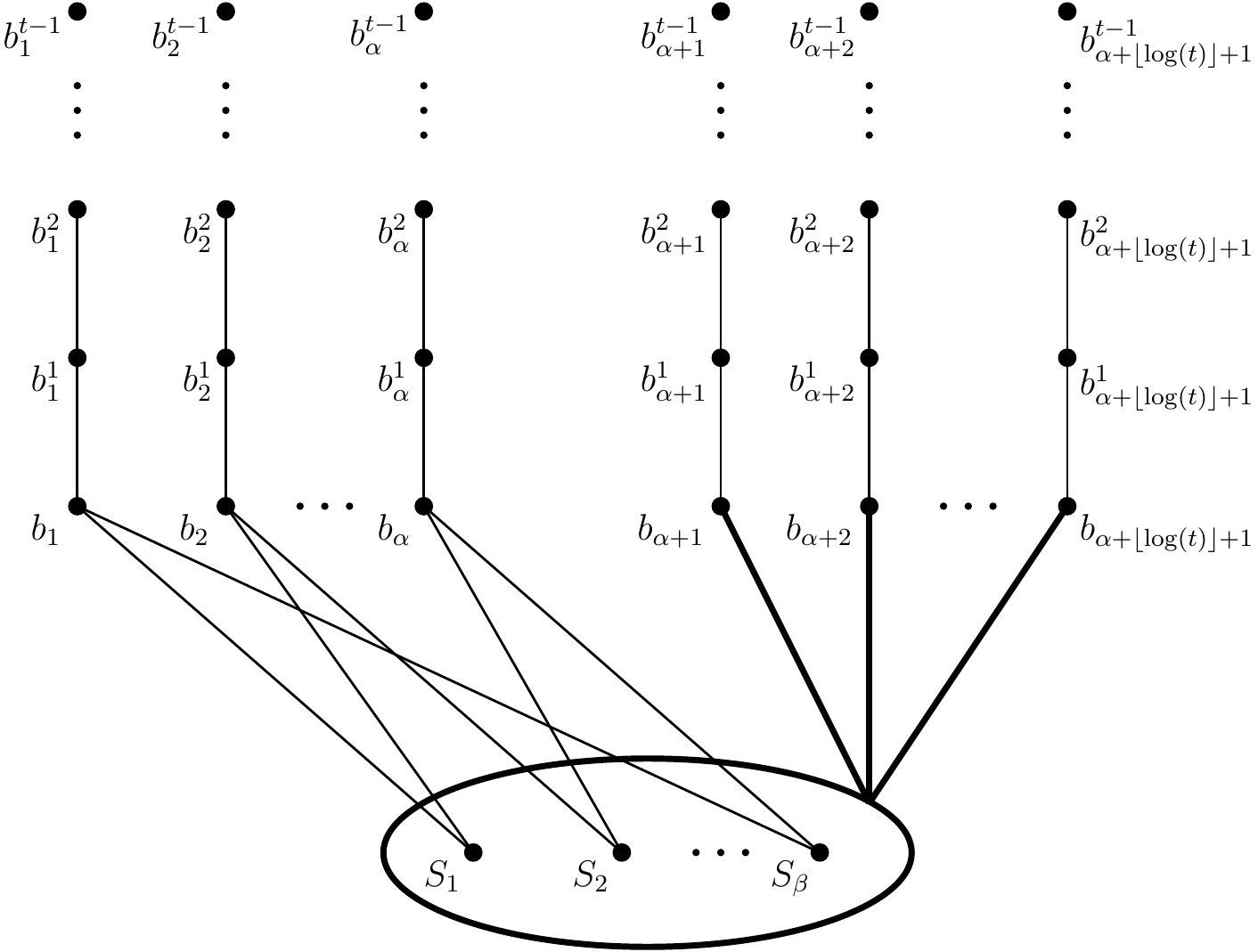}
\caption{The construction of $G$ from an instance of set cover. The vertices inside the oval form a clique. The bolded edges to the clique indicate adjacency to all of the vertices in the clique.}\label{fig:np-hard}
\end{figure}

We first prove the forward direction. In particular, we show that if there exists a set cover of size at most $k$, then $c_t^{\infty}(G)\leq k+\left\lfloor \log{t} \right\rfloor+1$. Assume there exists a set cover of size at most $k$. Note that, by the construction of $G$, any subset vertex covers the $\left\lfloor \log{t} \right\rfloor+1$ additional paths of size $t$. Place a cop on each of the subset vertices corresponding to such a set cover. Place the remaining at least $\left\lfloor \log{t} \right\rfloor+1$ cops on any subset vertex. Note that the cops can maintain that they occupy the subset vertices that correspond to a set cover as long as at least $k$ cops occupy the subset vertices since they form a clique. It is straightforward to see that as long as the cops occupy the subset vertices corresponding to a set cover, they can capture the robber in at most $t$ time-steps, no matter where he places himself. One of the closest cops simply follows a shortest path to the robber. Note that this cop is at distance at most $t$ from the robber since the cops occupy a set cover and thus, if the robber does not occupy a subset vertex (in which case he is immediately captured), then the cop moves to an element vertex, and so either the robber occupies that element vertex and is thus, captured, or the robber occupies the path attached to that element vertex and therefore, the cop separates the robber from the rest of the graph.

The cops follow this strategy to capture the robber. When a cop leaves a subset vertex, the remaining cops in the subset vertex clique move so that they occupy a set cover again. When a cop that does not occupy a subset vertex has captured the robber, that cop moves back to the subset clique. Note that this takes at most $t$ time-steps. Therefore, the robber can only win if he manages to force $\left\lfloor \log{t} \right\rfloor+2$ cops out of the subset clique (or equivalently, force $\left\lfloor \log{t} \right\rfloor+1$ cops out of the subset clique and force one cop to occupy a subset vertex such that no set cover of size at most $k$ can be formed including this subset).

By Lemma~\ref{lem:sequence} and the cops' strategy, this is not possible. To illustrate why this argument is valid, let us consider the case where $t$ is a power of $2$. By Lemma~\ref{lem:sequence}, at best, the robber could attack the attached paths and force himself to be captured in $t$ time-steps, then $\frac{t}{2}$, then $\frac{t}{4}$ time-steps, all the way down to one time-step. From Lemma~\ref{lem:sequence}, we find that this would only force $\left\lfloor \log{t} \right\rfloor+1$ cops out of the subset clique, which is not enough as mentioned in the paragraph above.

For the reverse direction, we show that if there is no set cover of size at most $k$, then $c_t^{\infty}(G)> k+\left\lfloor \log{t} \right\rfloor+1$. The robber places himself at a leaf of one of the additional $\lfloor \log{t} \rfloor + 1$ paths of size $t$ and passes until he is captured by a cop. The second robber places himself at the vertex distance $2^{\left\lfloor \log{t} \right\rfloor-1}-1$ from the element vertex in another one of these additional paths that contains no cop at a vertex distance greater than or equal to $2^{\left\lfloor \log{t} \right\rfloor-1}-1$ from the element vertex in the same path. Note that if no such additional path exists, then there are at most $k$ cops occupying vertices in the subset clique and the robber wins by placing himself at the end of an attached path such that its element is not covered by any of the $k$ cops in the subset clique. This exists since there is no set cover of size at most $k$ by assumption. The robber wins in this case, since there is no cop within distance $t$ of him. Thus, we can assume that such an additional path exists. The robber can easily force a cop to capture him in at most $2^{\left\lfloor \log{t} \right\rfloor-1}$ time-steps at his initial position. If he passes his first $2^{\left\lfloor \log{t} \right\rfloor-1}-1$ turns and no cop captures him, then the robber can just move to the leaf of the path, guaranteeing being captured in at least $t+1$ time-steps and so winning. Thus, assume that the second robber is captured at his initial position after at most $2^{\left\lfloor \log{t} \right\rfloor-1}$ time-steps in this play. In general, for $1\leq i \leq \left\lfloor \log{t} \right\rfloor +1$, the $i^{th}$ robber places himself at the vertex distance $2^{\left\lfloor \log{t} \right\rfloor-i+1}-1$ from the element vertex in another one of these additional paths that contains no cop at a vertex distance greater than or equal to his position from the element vertex in the same path. Again, as above, if no such path exists, then the robber wins and so, we assume that such a path exists. The $i^{th}$ robber forces himself to be captured at his initial position in $2^{\left\lfloor \log{t} \right\rfloor-i+1}$ time-steps in the $i^{th}$ play. By Lemma~\ref{lem:sequence}, after the $(\left\lfloor \log{t} \right\rfloor + 1)^{th}$ robber has been captured, there are at least $\left\lfloor \log{t} \right\rfloor + 1$ cops not in the subset clique. If there are at most $k+\left\lfloor \log{t} \right\rfloor+1$ cops, then the next robber wins by placing himself at the end of an attached path such that its element is not covered by any of the at most $k$ cops in the subset clique. This exists since there is no set cover of size at most $k$ by assumption. The robber wins in this case, since there is no cop within distance $t$ of him. Hence, $c_t^{\infty}(G)> k+\left\lfloor \log{t} \right\rfloor+1$.
\end{proof}

We finish the section by noting that the eternal cop number can be computed in polynomial time when $k$ is fixed.

\begin{theorem}
Given a graph $G$, $t\in \mathbb{N^*}$, and a fixed constant $k\in \mathbb{N^*}$, the problem of deciding whether $c_t^{\infty}(G)\leq k$ is polynomial-time solvable.
\end{theorem}

\begin{proof}
There are at most $n\binom{n+k-1}{k}=O(n^{k+1})$ possible configurations of the cops and the robber, and, since $k$ is fixed, the configurations graph is therefore of polynomial size.
\end{proof}

\section{The eternal cop number of paths and cycles}\label{sec:paths&cycles}

Since computing the eternal cop number of general graphs is \textbf{NP}-hard, we focus on its value for various graph families.  In this section, we give precise values of the eternal cop number on paths and cycles.
\begin{theorem} \label{thm:paths}
For all $t\in \mathbb{N^*}$, we have that $c_t^{\infty}(P_n)=\left\lceil\frac{n}{t+1}\right\rceil$.
\end{theorem}

\begin{proof}
Let the vertices of the path $P_n$ be $v_0,v_1,\ldots,v_{n-1}$. For the lower bound, the first robber places himself at $v_0$ and then passes his turn until he is captured.
There must be a cop $C_1$ to capture the first robber.
Once the first robber is captured, the second robber places himself at $v_{t+1}$ and passes his turn until he is captured.
Since $v_0$ is at distance $t+1$ from $v_{t+1}$ and the robber is not moving, $C_1$ cannot capture the robber in this play, so there must be another cop $C_2$ that eventually moves to $v_{t+1}$ to capture the robber.
Continuing in this fashion, for $0\leq i \leq \left\lceil\frac{n}{t+1}\right\rceil-1$, the $i^{th}$ robber's strategy is to place himself at $v_{i(t+1)}$ and remain there until captured, which forces there to be $\left\lceil\frac{n}{t+1}\right\rceil$ cops.

For the upper bound, for every $1\leq i\leq \left\lceil\frac{n}{t+1}\right\rceil$, assign a cop $C_i$ to the vertices in the set $\{v_{(i-1)(t+1)}, \ldots,v_{i(t+1)-1} \}$. A cop never leaves her assigned subpath of diameter at most $t$ and all cops move on a shortest path to the robber on their turn. If the robber remains in the same subpath, then he will be captured in at most $t$ time-steps since the diameter of the subpath is at most $t$. If the robber leaves to another subpath, he will be captured in at most $t$ time-steps since all of the cops are moving towards the robber without leaving their subpaths of diameter at most $t$.
\end{proof}

The results for cycles are separated into three theorems, with the first one covering the case of small cycles.

\begin{theorem}
Let $k\in \{4,5,6\}$. For all $t \in \mathbb{N^*}$, we have that $c_t^{\infty}(C_k)=2$.
\end{theorem}

\begin{proof}
For all $k\in \{4,5,6\}$, we have that $2=c(C_k)\leq c_t^{\infty}(C_k) \leq \gamma^{\infty}_{\mathrm{all}}(C_k)=\left\lceil \frac{k}{3} \right\rceil=2$.
\end{proof}

The next theorem covers the case of cycles when $t$ is large (around at least half the vertices).

\begin{theorem}
For $n\geq 7$ and for all integers $t \geq \left\lceil \frac{n}{2} \right\rceil-2$, we have that $c_t^{\infty}(C_n)=2$.
\end{theorem}

\begin{proof}
At least two cops are necessary since $c(C_n)=2$. Suppose we have two cops, and assume the cops are initially at a distance of at least three from each other, which is possible since $n \geq 7$. When the robber appears, each cop chases the robber along the unobstructed (not containing the other cop) subpath of vertices between herself and the robber.
Supposing that the robber is not captured during these moves, the cops continue making these moves until they are at distance $3$ or $4$ from each other. This will occur within at most $t-1$ time-steps for $n$ even, and within at most $t-2$ time-steps for $n$ odd, by the assumption that $t \geq \left\lceil \frac{n}{2} \right\rceil-2$.
It is straightforward to check that the cops can now capture the robber in either 1 move (distance 3 apart) or $2$ moves (distance 4 apart) while maintaining a distance of at least $3$ between the cops.
The condition that the two cops are at distance at least $3$ from each other therefore holds when the robber appears, and after the robber is captured. Hence, the cops can capture the robber in every play.
\end{proof}

The last theorem for cycles deals with the case when $t$ is ``small"; that is, approximately at most half the vertices. A set $S\subseteq V(G)$ \emph{distance $t$-dominates $G$} if for each vertex $v\in V(G)$, there exists a vertex $u\in S$ such that $dist(v,u)\leq t$. In the following theorem, the cops' strategy is to $t$-distance dominate the graph at all times.

\begin{theorem}
For $n\geq 7$ and for all integers $1 \leq t \leq \left\lceil \frac{n}{2} \right\rceil-3$, we have that $c_t^{\infty}(C_n)=\left\lceil\frac{n-3}{2t+1}\right\rceil+1$.
\end{theorem}

\begin{proof}
First, we prove the lower bound. We will show the robber wins against $\left\lceil\frac{n-3}{2t+1}\right\rceil$ cops. The robber places himself at any vertex non-adjacent to a cop if possible and if not, then there are at least two cops at distance at most 3 from each other. In either case, there exists a robber strategy such that, after a certain number of moves, either the cops have lost or a cop captures the robber and there is at least one other cop at distance at most 3 from this cop at the end of the cops' turn. The first robber employs this strategy.

Now these two cops that are distance at most 3 from each other, can only reach a vertex at distance at most $t$ from their current position in the next play. Therefore, they can reach at most $2t+4$ vertices in the next play. The remainder of the cops must distance $t$-dominate the cycle or else the second robber places himself at a vertex at distance at least $t+1$ from all of the cops and passes his turn until he wins. To distance $t$-dominate the rest of the cycle, at least $\left\lceil\frac{n-2t-4}{2t+1}\right\rceil=\left\lceil\frac{n-3}{2t+1}\right\rceil-1$ additional cops are required. That makes for a total of $\left\lceil\frac{n-3}{2t+1}\right\rceil+1$ cops but there are only $\left\lceil\frac{n-3}{2t+1}\right\rceil$ cops, therefore, the robber wins.

Now, we prove the upper bound. Since $c_t^{\infty}(C_n)$ monotonically increases with $n$, we can assume that $\frac{n-3}{2t+1}\in \mathbb{N}$, because if $k$ cops win on a cycle of length $p$, then $k$ cops will certainly win on a cycle of length $m$ where $m < p$.

Start the cops in any arbitrary configuration such that the distance between any two consecutive cops is at most $2t+1$ and at least 3 (any subpath of the cycle containing no cops has size at most $2t$ and at least 2) and no vertex is occupied by more than one cop. Such a configuration is possible since there are $\frac{n-3}{2t+1}+1$ cops and $t\geq 1$. Now, we will show that the cops can capture the robber in at most $t$ time-steps and at the end of the play, the cops are again in a configuration where the distance between any two consecutive cops is at most $2t+1$ and at least 3, and hence, the cops can always capture the robber in at most $t$ time-steps in any play.

Assume we are in the first play. The robber places himself at some non-occupied vertex (as otherwise he is already captured). The closest cop in the clockwise direction towards the robber and the closest cop in the counterclockwise direction towards the robber, move on a shortest path to the robber, while maintaining that they are at distance at least 3 from each other. One of these cops will capture the robber in at most $t$ time-steps since they are distance at most $2t+1$ apart.

Let $k=\frac{n-3}{2t+1}+1$ and let the closest cop in the clockwise (counterclockwise, respectively) direction towards the robber be $C_1$ ($C_k,$ respectively). For an integer $1 < j < k$ ($1 < \ell < k,$ respectively), let $C_j$ ($C_{\ell},$ respectively) be the first cop that is distance strictly less than $2t+1$ from the nearest cop in the clockwise (counterclockwise, respectively) direction from herself, starting from $C_1$ ($C_k$, respectively), on turn $1 \leq s \leq t$ of the first play.

If neither $C_j$ nor $C_{\ell}$ exists, then the cops occupy a configuration in which $C_1$ and $C_k$ are distance 3 apart and the distance between any two other consecutive cops is exactly $2t+1$. In this case, on their next turn, the cops all move in the same direction (clockwise or counterclockwise) such that the robber is captured. Thus, the cops would still be in a configuration such that the distance between any two consecutive cops is at most $2t+1$ and at least 3 and no vertex is occupied by more than one cop. We therefore assume this is not the case.

The remainder of the cops utilize the following strategy. At time-step $s$, for all $1 < i < j$, the cop $C_i$ moves clockwise. Also, at time-step $s$, for all $\ell < h < k$, the cop $C_h$ moves counterclockwise. The rest of the cops (of course excluding $C_1$ and $C_k$ since their strategies have already been given) skip their turn at time-step $s$.

All that is left to be shown is that a cop is not asked to move both clockwise and counterclockwise in the strategy given above, that is, that either $C_j$ and $C_{\ell}$ exist (even if $j=\ell$) or that only one of $C_j$ and $C_{\ell}$, say $C_j$, exists, which implies that $C_j=C_2$ (since otherwise $C_{\ell}=C_{j-1}$). Note that we have already dealt with the case that neither of $C_j$ and $C_{\ell}$ exist and so, at least one of them must exist and thus, we are done.
\end{proof}

\section{Retracts and the eternal cop number}\label{sec:retracts}

We explore techniques that decompose a graph into particular subgraphs where it is simpler to analyze the eternal cop number. We apply these techniques to three graph classes: trees, certain strong products of graphs, and the Cartesian grid.

\subsection{Retracts}

A \emph{retract} $H$ of a graph $G$ is a subgraph of $G$ such that there exists a homomorphism $f:G\rightarrow H$, where $f(v)=v$ for all $v\in V(H)$. An \emph{isometric path} $P$ in a graph $G$ satisfies $d_H(v,w)=d_G(v,w)$ for all $v,w \in V(P)$. Note that any isometric path is a retract of $G$. Retracts have been useful in analyzing the cop number in a variety of contexts, including in planar and cop-win graphs; for more on retracts and isometric paths in Cops and Robbers, see \cite{BN}. The idea behind the \emph{shadow strategy} is that the cops play in the retract $H$ and capture $f(R).$ If the robber enters $H$, then they are captured.  We refer to $f(R)$ as the \emph{shadow} of the robber.

Suppose $G = H_1 \cup \ldots \cup H_p$ is the disjoint union of retracts. We may place a sufficient number of cops to capture the robber in each retract $H$ in the decomposition. Once all shadows are captured, the robber is captured as $f(R)=R$ for exactly one retract $H$; hence, this approach gives an upper bound on the cop number of $G$. An analogous approach is taken when considering the \emph{isometric path number} or \emph{precinct number} \cite{BN}, where graphs are decomposed into a minimum set of isometric paths.

Our next lemma follows immediately from using the retract decomposition technique in the setting of Cops and Eternal Robbers.
\begin{lemma} \label{lem:disjointRetracts}
For all $t\in \mathbb{N^*}$, if $G = H_1 \cup \ldots \cup H_p$ is the union of retracts, then
\[
c_t^\infty(G) \leq \sum_{i=1}^p c_t^\infty(H_i).
\]
\end{lemma}

\subsection{A Recurrent Attack Strategy}

In (\ref{xx}), if the capture time on $G$ with $k$ cops is at most $t/2$, then $c_t^\infty(G) \leq k$. This observation used what we refer to as an \emph{attack phase} of $t/2$ time-steps where the robber is captured by the $k$ cops, followed by a \emph{regroup phase} of $t/2$ time-steps where the cops return to their initial configuration. We generalize the idea behind this observation by placing more cops, and hence, allowing some cops to attack while other cops regroup.

To give an initial demonstration of this idea, consider a spider graph with three branches of length four, with two cops and $t=5$. Figures~2A~-~2D demonstrate a sequence of plays with two cops and one robber on a spider with three legs of length $4$, where $t=5$.
It can be shown that two cops win on the graph $G$ in this example, and so $c^\infty_5(G)=2$. 
Note that $c_5(G)=1$. Our example demonstrates the regrouping phase and attacking phases of play, which will be formalized in the corresponding theorem, Theorem~\ref{thm:infinity_from_timed}.
The cops are denoted by $C_{1}$ and $C_{2}$, and the robber is denoted by $R$.

In play 1 (see Figure~2A), the robber appears on a leaf.
Cop $C_1$ captures the robber in four steps.
In play 2, one time-step is spent regrouping (Figure~2B), and then cop $C_2$ attacks the robber (Figure~2C).
At the start of play 3 (Figure~2D), the cops are positioned such that $C_1$ can capture the robber while $C_2$ regroups back to the root. The robber can always be captured in this fashion, and in every subsequent play, the robber will likewise be captured.

If the graph was extended by appending one vertex to a leaf (Figure~2E), then neither cop could capture the robber in $t=5$ time-steps. This fact is independent of how the cops chose to play; the particular cop strategy we have displayed in this example can be shown to be optimal.
Thus, on the extended graph $G'$, $c^\infty_5(G')>2$.

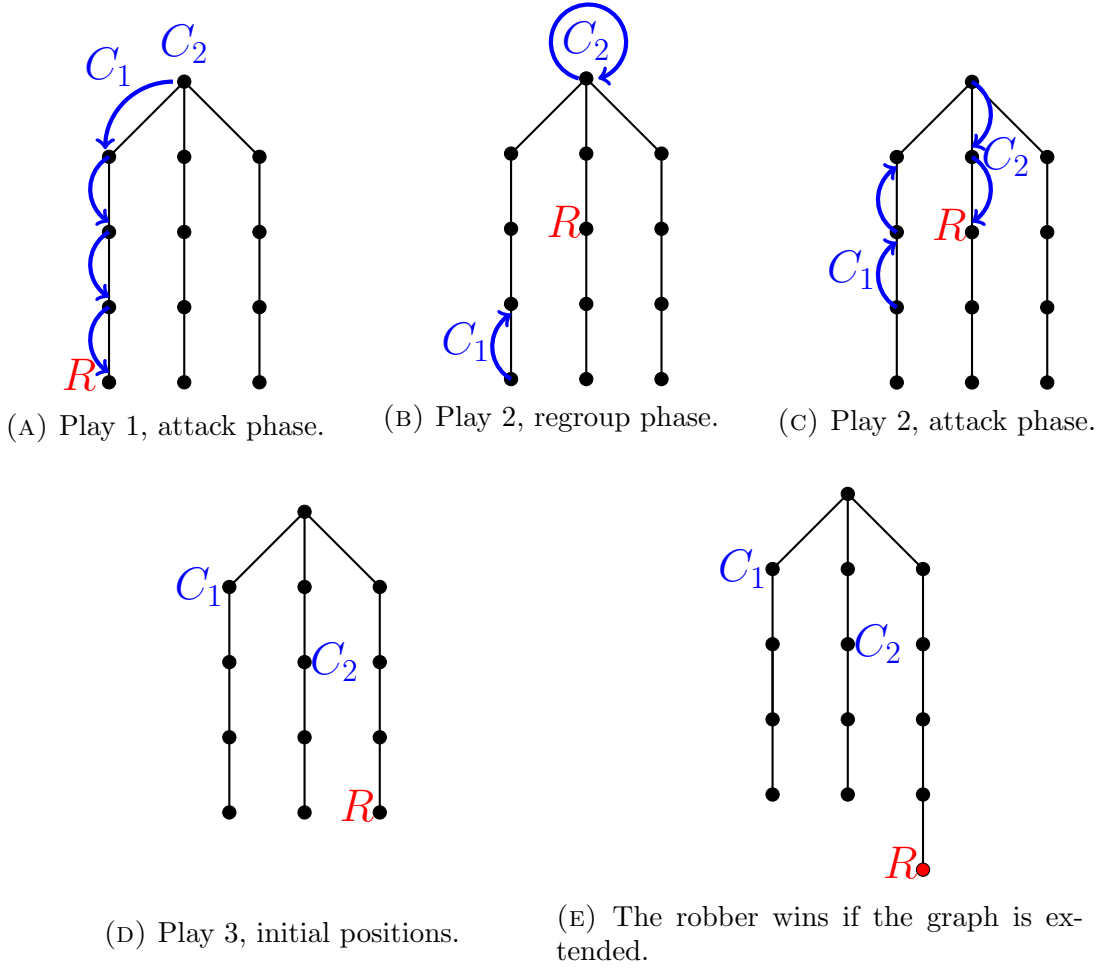
\begin{figure}
  \vspace{10pt}
\begin{subfigure}{.3\textwidth}
  \centering
  \begin{tikzpicture}
\draw [line width=0.8pt] (1,4)-- (0,3);
\draw [line width=0.8pt] (0,2)-- (0,3);
\draw [line width=0.8pt] (0,2)-- (0,1);
\draw [line width=0.8pt] (0,1)-- (0,0);
\draw [line width=0.8pt] (1,4)-- (1,3);
\draw [line width=0.8pt] (1,2)-- (1,3);
\draw [line width=0.8pt] (1,2)-- (1,1);
\draw [line width=0.8pt] (1,1)-- (1,0);
\draw [line width=0.8pt] (1,4)-- (2,3);
\draw [line width=0.8pt] (2,2)-- (2,3);
\draw [line width=0.8pt] (2,2)-- (2,1);
\draw [line width=0.8pt] (2,1)-- (2,0);
\begin{scriptsize}
\draw [fill=black] (1,4) circle (2.5pt);
\draw [fill=black] (0,3) circle (2.5pt);
\draw [fill=black] (0,2) circle (2.5pt);
\draw [fill=black] (0,1) circle (2.5pt);
\draw [fill=black] (0,0) circle (2.5pt);
\draw [fill=black] (1,3) circle (2.5pt);
\draw [fill=black] (1,2) circle (2.5pt);
\draw [fill=black] (1,1) circle (2.5pt);
\draw [fill=black] (1,0) circle (2.5pt);
\draw [fill=black] (2,3) circle (2.5pt);
\draw [fill=black] (2,2) circle (2.5pt);
\draw [fill=black] (2,1) circle (2.5pt);
\draw [fill=black] (2,0) circle (2.5pt);
\draw[color=blue] (1,4.5) node {\Large $C_{2}$};
\draw[color=blue] (0, 4.2) node {\Large $C_{1}$};
\draw[color=red] (-0.4,0.1) node {\Large $R$};
\draw[color=blue,ultra thick, ->] (0.85,4) arc (90:180:0.9);
\draw[color=blue,ultra thick, ->] (0,1) arc (120:240:0.5);
\draw[color=blue,ultra thick, ->] (0,2) arc (120:240:0.5);
\draw[color=blue,ultra thick, ->] (0,3) arc (120:240:0.5);
\end{scriptsize}
\end{tikzpicture}
\vspace{-0.05cm}
  \caption{
  Play 1, attack phase.}
  \label{fig:sfig1}
\end{subfigure}
\hspace{5pt}
\begin{subfigure}{.3\textwidth}
  \centering
  \vspace{-0.36cm}
  \begin{tikzpicture}
\draw [line width=0.8pt] (1,4)-- (0,3);
\draw [line width=0.8pt] (0,2)-- (0,3);
\draw [line width=0.8pt] (0,2)-- (0,1);
\draw [line width=0.8pt] (0,1)-- (0,0);
\draw [line width=0.8pt] (1,4)-- (1,3);
\draw [line width=0.8pt] (1,2)-- (1,3);
\draw [line width=0.8pt] (1,2)-- (1,1);
\draw [line width=0.8pt] (1,1)-- (1,0);
\draw [line width=0.8pt] (1,4)-- (2,3);
\draw [line width=0.8pt] (2,2)-- (2,3);
\draw [line width=0.8pt] (2,2)-- (2,1);
\draw [line width=0.8pt] (2,1)-- (2,0);
\begin{scriptsize}
\draw [fill=black] (1,4) circle (2.5pt);
\draw [fill=black] (0,3) circle (2.5pt);
\draw [fill=black] (0,2) circle (2.5pt);
\draw [fill=black] (0,1) circle (2.5pt);
\draw [fill=black] (0,0) circle (2.5pt);
\draw [fill=black] (1,3) circle (2.5pt);
\draw [fill=black] (1,2) circle (2.5pt);
\draw [fill=black] (1,1) circle (2.5pt);
\draw [fill=black] (1,0) circle (2.5pt);
\draw [fill=black] (2,3) circle (2.5pt);
\draw [fill=black] (2,2) circle (2.5pt);
\draw [fill=black] (2,1) circle (2.5pt);
\draw [fill=black] (2,0) circle (2.5pt);
\draw[color=blue] (1,4.5) node {\Large $C_{2}$};
\draw[color=red] (0.7, 2.1) node {\Large $R$};
\draw[color=blue] (-0.6,0.5) node {\Large $C_{1}$};
\draw[color=blue,ultra thick, ->] (0.9,4) arc (255:-75:0.5);
\draw[color=blue,ultra thick, ->] (0,0) arc (240:120:0.5);
\end{scriptsize}
\end{tikzpicture}
\caption{
  Play 2, regroup phase.}
  \label{fig:sfig2}
\end{subfigure}
\hspace{5pt}
\begin{subfigure}{.3\textwidth}
  \centering
  \vspace{-0.32cm}
  \begin{tikzpicture}
\draw [line width=0.8pt] (1,4)-- (0,3);
\draw [line width=0.8pt] (0,2)-- (0,3);
\draw [line width=0.8pt] (0,2)-- (0,1);
\draw [line width=0.8pt] (0,1)-- (0,0);
\draw [line width=0.8pt] (1,4)-- (1,3);
\draw [line width=0.8pt] (1,2)-- (1,3);
\draw [line width=0.8pt] (1,2)-- (1,1);
\draw [line width=0.8pt] (1,1)-- (1,0);
\draw [line width=0.8pt] (1,4)-- (2,3);
\draw [line width=0.8pt] (2,2)-- (2,3);
\draw [line width=0.8pt] (2,2)-- (2,1);
\draw [line width=0.8pt] (2,1)-- (2,0);
\begin{scriptsize}
\draw [fill=black] (1,4) circle (2.5pt);
\draw [fill=black] (0,3) circle (2.5pt);
\draw [fill=black] (0,2) circle (2.5pt);
\draw [fill=black] (0,1) circle (2.5pt);
\draw [fill=black] (0,0) circle (2.5pt);
\draw [fill=black] (1,3) circle (2.5pt);
\draw [fill=black] (1,2) circle (2.5pt);
\draw [fill=black] (1,1) circle (2.5pt);
\draw [fill=black] (1,0) circle (2.5pt);
\draw [fill=black] (2,3) circle (2.5pt);
\draw [fill=black] (2,2) circle (2.5pt);
\draw [fill=black] (2,1) circle (2.5pt);
\draw [fill=black] (2,0) circle (2.5pt);
\draw[color=red] (0.3, 4.8) node {\Large \phantom{$R$}};
\draw[color=red] (0.7, 2.1) node {\Large $R$};
\draw[color=blue] (1.45,3) node {\Large $C_{2}$};
\draw[color=blue] (-0.6,1.5) node {\Large $C_{1}$};
\draw[color=blue,ultra thick, ->] (0,1) arc (240:120:0.5);
\draw[color=blue,ultra thick, ->] (0,2) arc (240:120:0.5);
\draw[color=blue,ultra thick, ->] (1,3) arc (60:-60:0.5);
\draw[color=blue,ultra thick, ->] (1,4) arc (60:-60:0.5);
\end{scriptsize}
\end{tikzpicture}
\caption{
  Play 2, attack phase.}
  \label{fig:sfig3}
\end{subfigure}
\begin{subfigure}{.45\textwidth}
  \centering
  \begin{tikzpicture}
\draw [line width=0.8pt] (1,4)-- (0,3);
\draw [line width=0.8pt] (0,2)-- (0,3);
\draw [line width=0.8pt] (0,2)-- (0,1);
\draw [line width=0.8pt] (0,1)-- (0,0);
\draw [line width=0.8pt] (1,4)-- (1,3);
\draw [line width=0.8pt] (1,2)-- (1,3);
\draw [line width=0.8pt] (1,2)-- (1,1);
\draw [line width=0.8pt] (1,1)-- (1,0);
\draw [line width=0.8pt] (1,4)-- (2,3);
\draw [line width=0.8pt] (2,2)-- (2,3);
\draw [line width=0.8pt] (2,2)-- (2,1);
\draw [line width=0.8pt] (2,1)-- (2,0);
\begin{scriptsize}
\draw [fill=black] (1,4) circle (2.5pt);
\draw [fill=black] (0,3) circle (2.5pt);
\draw [fill=black] (0,2) circle (2.5pt);
\draw [fill=black] (0,1) circle (2.5pt);
\draw [fill=black] (0,0) circle (2.5pt);
\draw [fill=black] (1,3) circle (2.5pt);
\draw [fill=black] (1,2) circle (2.5pt);
\draw [fill=black] (1,1) circle (2.5pt);
\draw [fill=black] (1,0) circle (2.5pt);
\draw [fill=black] (2,3) circle (2.5pt);
\draw [fill=black] (2,2) circle (2.5pt);
\draw [fill=black] (2,1) circle (2.5pt);
\draw [fill=black] (2,0) circle (2.5pt);
\draw[color=red] (1.7, 0.1) node {\Large $R$};
\draw[color=red] (1.7, -1) node {\Large \phantom{$R$}};
\draw[color=red] (0.3, 4.3) node {\Large \phantom{$R$}};
\draw[color=blue] (1.4,2) node {\Large $C_{2}$};
\draw[color=blue] (-0.4,3) node {\Large $C_{1}$};

\end{scriptsize}
\end{tikzpicture}
\vspace{-0.12cm}
\caption{
  Play 3, initial positions.}
  \label{fig:sfig4}
\end{subfigure}
\begin{subfigure}{.45\textwidth}
  \centering
\begin{tikzpicture}
\draw [line width=0.8pt] (1,4)-- (0,3);
\draw [line width=0.8pt] (0,2)-- (0,3);
\draw [line width=0.8pt] (0,2)-- (0,1);
\draw [line width=0.8pt] (0,1)-- (0,0);
\draw [line width=0.8pt] (1,4)-- (1,3);
\draw [line width=0.8pt] (1,2)-- (1,3);
\draw [line width=0.8pt] (1,2)-- (1,1);
\draw [line width=0.8pt] (1,1)-- (1,0);
\draw [line width=0.8pt] (0,2)-- (0,1);
\draw [line width=0.8pt] (1,4)-- (2,3);
\draw [line width=0.8pt] (2,2)-- (2,3);
\draw [line width=0.8pt] (2,2)-- (2,1);
\draw [line width=0.8pt] (2,1)-- (2,0);
\draw [line width=0.8pt] (2,0)-- (2,-1);
\begin{scriptsize}
\draw [fill=black] (1,4) circle (2.5pt);
\draw [fill=black] (0,3) circle (2.5pt);
\draw [fill=black] (0,2) circle (2.5pt);
\draw [fill=black] (0,1) circle (2.5pt);
\draw [fill=black] (0,0) circle (2.5pt);
\draw [fill=black] (1,3) circle (2.5pt);
\draw [fill=black] (1,2) circle (2.5pt);
\draw [fill=black] (1,1) circle (2.5pt);
\draw [fill=black] (1,0) circle (2.5pt);
\draw [fill=black] (2,3) circle (2.5pt);
\draw [fill=black] (2,2) circle (2.5pt);
\draw [fill=black] (2,1) circle (2.5pt);
\draw [fill=black] (2,0) circle (2.5pt);
\draw [fill=red] (2,-1) circle (2.5pt);
\draw[color=red] (1.7, -0.9) node {\Large $R$};
\draw[color=blue] (1.4,2) node {\Large $C_{2}$};
\draw[color=blue] (-0.4,3) node {\Large $C_{1}$};
\draw[color=red] (0.3, 4.3) node {\Large \phantom{$R$}};
\end{scriptsize}
\end{tikzpicture}
\caption{
The robber wins if the graph is extended.}\label{fig:robber_escape}
\end{subfigure}
\caption{Recurrent attack strategy for the cops on a spider.}
\label{fig:fig}
\end{figure}

For all $i\in \mathbb{N^*}$, we define $\mathcal{G}_i^k$ to be the set of all graphs in which $k\in \mathbb{N^*}$ cops are sufficient to capture a robber in $\ell_i = \lceil (1-2^{-i})t -1/2 \rceil$
time-steps; that is, $\mathcal{G}_i^k$ consists of all graphs $G$ with $c_{\ell_i}(G) \leq k$.
We note that both $\ell_i$ and $\mathcal{G}_i^k$ are dependent on the value of $t$. However, we do not include this explicitly in the notation. Observe that $t/2-1/2 \leq \ell_i \leq t$ for all $i\geq 1$.

\begin{theorem} \label{thm:infinity_from_timed}
For all $t,i,k\in \mathbb{N^*}$ and all graphs $G\in \mathcal{G}_i^k$, $c_t^\infty(G) \leq i k$.
\end{theorem}

\begin{proof}
Since $G\in \mathcal{G}_i^k$, then $c_{\ell_i}(G)\leq k$ by definition. Suppose we are playing with $ik$ cops. For convenience, let $\delta=t-\ell_i$.

During the first play, $c_{\ell_i}(G)\leq k$ cops are used to capture the robber.
During the second play, a second set of $ c_{\ell_i}(G)\leq k$ cops is used to capture the robber.
In total, $i$ different sets of $c_{\ell_i}(G)\leq k$ cops will be required; the set of cops used in play $j$ will be referred to as \emph{unit} $j$.

Suppose that $(v_1, \ldots, v_k)$ (for all integers $1 \leq j \leq k$, $v_j\in V(G)$) is the initial configuration of the $k$ cops in a strategy that captures the robber in a single play in at most $\ell_i$ time-steps.
Initially, place $i$ cops at each vertex $v_j$, for $1 \leq j \leq k$, creating $i$ units, for a total of $ik$ cops.

The cop strategy is as follows. For each play:

\begin{enumerate}
\item \emph{Regroup phase}: During the first $\delta$ rounds, all cops move on a shortest path to their initial starting position in the initial configuration.
\item \emph{Attack phase}: Some unit $u_j$ that is in its initial configuration is selected as the attacking unit.
For the next $\ell_i = t-\delta$ rounds, the attacking unit captures the robber, while all the other units return to their initial configurations.
\end{enumerate}

The strategy above can only fail if there is no unit that has completely returned back to its initial configuration by the start of the attack phase. For a contradiction, assume that the robber wins say during play $P_f,$ where $f$ represents the final play.
This means that at the start of this play, it would have taken each of the units at least $\delta+1$ time-steps to return to their initial configuration (as from the initial configuration, the robber could be caught within $t-\delta$ time-steps by any unit in its initial configuration).

The previous play (say, play $P_{f-1}$) ended when a unit, say unit $u_1$, captured a robber. By the assumptions above, this implies that in play $P_{f-1}$, unit $u_1$ took at least $\delta+1$ rounds to capture the robber during the attack phase.
It then follows that it would take each of the other units at least $2(\delta+1)$ time-steps to return to their initial configuration at the start of the attack phase, and hence, it would take each of the other units at least $2(\delta+1)+\delta = 3\delta + 2$ time-steps to return to their initial configuration at the start of the play.

Continuing in this fashion, on play $P_{f-i}$, unit $u_i$ needs to have captured a robber such that it would take at least $(2^i -1) \delta + (2^{i-1})$ time-steps for the unit $u_i$ to return to its initial configuration. If this occurred, then $\ell_i \geq (2^i -1) \delta + (2^{i-1})$, which when expanded for $\delta=t-\ell_i$ and rearranged becomes $\ell_i \geq (1-2^{-i})t + \frac{1}{2}$, contradicting the definition of $\ell_i$.
Thus, our assumption was incorrect, and the robber does not win.
\end{proof}

In the next subsections, we will study various graph classes.
For part of this analysis, we will find retracts that exist as subgraphs within the graph classes and that are also in $\mathcal{G}_i^k$ for small values of $k$.
It will then be useful to apply Theorem~\ref{thm:infinity_from_timed} to the retracts rather than to the graph as a whole.

\begin{corollary} \label{cor:retracts_give_CER}
Let $t>0$ be a fixed integer.
Suppose we have a decomposition of $G$ into retracts, such that each retract $R$ in the decomposition can be assigned a parameter $(i,k)$ such that $R\in \mathcal{G}_i^k$.
Let $\alpha_i^k$ be the total number of retracts in the decomposition that have been assigned a parameter of $(i,k)$.
We then have that
\[
c_t^{\infty}(G) \leq \sum_{i,k} \alpha_i^k \cdot i \cdot k.
\]
\end{corollary}
We note that for each retract $R$, naturally we would want to select an $i$ and $k$ with $R\in \mathcal{G}_i^k$ such that $i \cdot k$ is minimized.

In Figure~\ref{fig:decomp_retracts}, we give a graph along with a decomposition into four retracts denoted using surrounding ellipses. When $t=4$, the first retract (denoted by a surrounding blue circle) is in  $\mathcal{G}_1^2$ and so we select $(i,k)=(1,2)$,  the second retract (denoted by a surrounding red circle) is in  $\mathcal{G}_4^1$ and so we select $(i,k)=(4,1)$, and the third and fourth (denoted by surrounding green circles) are in  $\mathcal{G}_1^1$ and so we select $(i,k)=(1,1)$. Hence, we have that $\alpha_1^2=1$, $\alpha_4^1=1$, and $\alpha_1^1=2$. Corollary~\ref{cor:retracts_give_CER} yields that $c_t^{\infty}(G) \leq (\alpha_1^2 \cdot 2) + (\alpha_4^1 \cdot 4) + (\alpha_1^1 \cdot 1) = 8$.

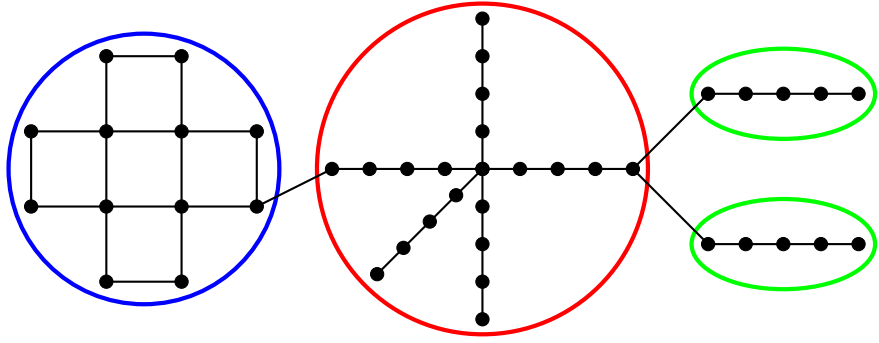
\begin{figure}
  \centering
  \begin{tikzpicture}
\draw [line width=0.8pt] (1,0)-- (4,0);
\draw [line width=0.8pt] (1,1)-- (4,1);
\draw [line width=0.8pt] (2,2)-- (3,2);
\draw [line width=0.8pt] (2,-1)-- (3,-1);
\draw [line width=0.8pt] (1,1)-- (1,0);
\draw [line width=0.8pt] (2,2)-- (2,-1);
\draw [line width=0.8pt] (3,2)-- (3,-1);
\draw [line width=0.8pt] (4,1)-- (4,0);
\begin{scriptsize}
\draw [fill=black] (1,0) circle (2.5pt);
\draw [fill=black] (1,1) circle (2.5pt);
\draw [fill=black] (2,0) circle (2.5pt);
\draw [fill=black] (2,1) circle (2.5pt);
\draw [fill=black] (3,0) circle (2.5pt);
\draw [fill=black] (3,1) circle (2.5pt);
\draw [fill=black] (4,0) circle (2.5pt);
\draw [fill=black] (4,1) circle (2.5pt);
\draw [fill=black] (3,-1) circle (2.5pt);
\draw [fill=black] (3,2) circle (2.5pt);
\draw [fill=black] (2,-1) circle (2.5pt);
\draw [fill=black] (2,2) circle (2.5pt);


\draw[color=blue,ultra thick] (2.5,0.5) ellipse (1.8cm and 1.8cm);
\end{scriptsize}

\draw [line width=0.8pt] (5,0.5)-- (9,0.5);
\draw [line width=0.8pt] (7,2.5)-- (7,-1.5);
\draw [line width=0.8pt] (7,0.5)-- (7-1.4,0.5-1.4);
\begin{scriptsize}
\draw [fill=black] (5,0.5) circle (2.5pt);
\draw [fill=black] (5.5,0.5) circle (2.5pt);
\draw [fill=black] (6,0.5) circle (2.5pt);
\draw [fill=black] (6.5,0.5) circle (2.5pt);
\draw [fill=black] (7,0.5) circle (2.5pt);
\draw [fill=black] (7.5,0.5) circle (2.5pt);
\draw [fill=black] (8,0.5) circle (2.5pt);
\draw [fill=black] (8.5,0.5) circle (2.5pt);
\draw [fill=black] (9,0.5) circle (2.5pt);
\draw [fill=black] (7,2.5) circle (2.5pt);
\draw [fill=black] (7,2) circle (2.5pt);
\draw [fill=black] (7,1.5) circle (2.5pt);
\draw [fill=black] (7,1) circle (2.5pt);
\draw [fill=black] (7,0.5) circle (2.5pt);
\draw [fill=black] (7,0) circle (2.5pt);
\draw [fill=black] (7,-0.5) circle (2.5pt);
\draw [fill=black] (7,-1) circle (2.5pt);
\draw [fill=black] (7,-1.5) circle (2.5pt);
\draw [fill=black] (7-0.35,0.5-0.35) circle (2.5pt);
\draw [fill=black] (7-0.7,0.5-0.7) circle (2.5pt);
\draw [fill=black] (7-1.05,0.5-1.05) circle (2.5pt);
\draw [fill=black] (7-1.4,0.5-1.4) circle (2.5pt);

\draw[color=red,ultra thick] (5-0.2,0.5) arc (-180:180:2.2);

\end{scriptsize}

\draw [line width=0.8pt] (10,1.5)-- (12,1.5);
\draw [line width=0.8pt] (10,-0.5)-- (12,-0.5);

\begin{scriptsize}
\draw [fill=black] (10,1.5) circle (2.5pt);
\draw [fill=black] (10.5,1.5) circle (2.5pt);
\draw [fill=black] (11,1.5) circle (2.5pt);
\draw [fill=black] (11.5,1.5) circle (2.5pt);
\draw [fill=black] (12,1.5) circle (2.5pt);
\draw[color=green,ultra thick] (11,1.5) ellipse (1.22cm and 0.6cm);

\draw [fill=black] (10,-0.5) circle (2.5pt);
\draw [fill=black] (10.5,-0.5) circle (2.5pt);
\draw [fill=black] (11,-0.5) circle (2.5pt);
\draw [fill=black] (11.5,-0.5) circle (2.5pt);
\draw [fill=black] (12,-0.5) circle (2.5pt);
\draw[color=green,ultra thick] (11,-0.5) ellipse (1.22cm and 0.6cm);
\end{scriptsize}

\draw [line width=0.8pt] (5,0.5)-- (4,0);
\draw [line width=0.8pt] (9,0.5)-- (9,0.5);
\draw [line width=0.8pt] (9,0.5)-- (10,1.5);
\draw [line width=0.8pt] (9,0.5)-- (10,-0.5);

\end{tikzpicture}
\vspace{-0.05cm}
  \caption{
  A graph decomposed into four retracts.
  }
  \label{fig:decomp_retracts}
\end{figure}

\subsection{Trees}\label{sec:trees}
A connected subtree of a tree is a retract. Further, every tree $T$ of radius $r \leq t$ has $c_r(T)=1$.
We note that the trees of radius at most $\ell_i = \lceil (1-2^{-i})t -1/2 \rceil$ are contained within $\mathcal{G}_i^1$.
We give examples of such trees in Figure~\ref{fig:examplesTrees}.
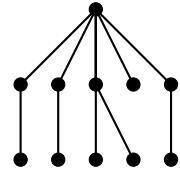
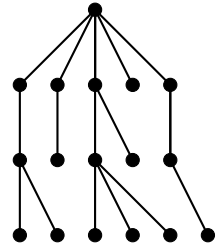
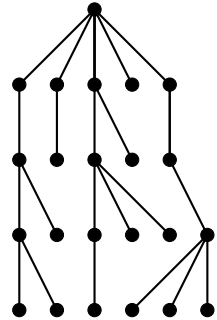
\begin{figure}
\begin{subfigure}{.31\textwidth}
  \vspace{1.98cm}
  \centering
  \begin{tikzpicture}
\draw [line width=0.8pt] (0,0)-- (-1,-1);
\draw [line width=0.8pt] (0,0)-- (0,-1);
\draw [line width=0.8pt] (0,0)-- (1,-1);
\draw [line width=0.8pt] (0,0)-- (-0.5,-1);
\draw [line width=0.8pt] (0,0)-- (0.5,-1);

\draw [line width=0.8pt] (-1,-2)-- (-1,-1);
\draw [line width=0.8pt] (-0.5,-2)-- (-0.5,-1);
\draw [line width=0.8pt] (0,0)-- (0,-1);
\draw [line width=0.8pt] (0,0)-- (0,-1);

\draw [line width=0.8pt] (0,-1)-- (0,-2);
\draw [line width=0.8pt] (0,-1)-- (0.5,-2);

\draw [line width=0.8pt] (1,-1)-- (1,-2);

\begin{scriptsize}
\draw [fill=black] (0,0) circle (2.5pt);
\draw [fill=black] (-1,-1) circle (2.5pt);
\draw [fill=black] (0,-1) circle (2.5pt);
\draw [fill=black] (1,-1) circle (2.5pt);
\draw [fill=black] (-0.5,-1) circle (2.5pt);
\draw [fill=black] (0.5,-1) circle (2.5pt);
\draw [fill=black] (-1,-2) circle (2.5pt);
\draw [fill=black] (0,-2) circle (2.5pt);
\draw [fill=black] (1,-2) circle (2.5pt);
\draw [fill=black] (-0.5,-2) circle (2.5pt);
\draw [fill=black] (0.5,-2) circle (2.5pt);
\end{scriptsize}
\end{tikzpicture}
\vspace{-0.05cm}
  \caption{
  $i=1$ and radius $\ell_1=2$.}
  \label{fig:treefig1}
\end{subfigure}

\hspace{5pt}
\begin{subfigure}{.31\textwidth}
  \vspace{1.38cm}
  \centering
  \begin{tikzpicture}
\draw [line width=0.8pt] (0,0)-- (-1,-1);
\draw [line width=0.8pt] (0,0)-- (0,-1);
\draw [line width=0.8pt] (0,0)-- (1,-1);
\draw [line width=0.8pt] (0,0)-- (-0.5,-1);
\draw [line width=0.8pt] (0,0)-- (0.5,-1);

\draw [line width=0.8pt] (-1,-2)-- (-1,-1);
\draw [line width=0.8pt] (-0.5,-2)-- (-0.5,-1);
\draw [line width=0.8pt] (0,0)-- (0,-1);
\draw [line width=0.8pt] (0,0)-- (0,-1);

\draw [line width=0.8pt] (0,-1)-- (0,-2);
\draw [line width=0.8pt] (0,-1)-- (0.5,-2);

\draw [line width=0.8pt] (1,-1)-- (1,-2);

\draw [line width=0.8pt] (1,-1)-- (1,-2);

\draw [line width=0.8pt] (-1,-2)-- (-1,-3);
\draw [line width=0.8pt] (-1,-2)-- (-0.5,-3);
\draw [line width=0.8pt] (0,-2)-- (0,-3);
\draw [line width=0.8pt] (0,-2)-- (0.5,-3);
\draw [line width=0.8pt] (0,-2)-- (1,-3);
\draw [line width=0.8pt] (1,-2)-- (1.5,-3);

\begin{scriptsize}
\draw [fill=black] (0,0) circle (2.5pt);
\draw [fill=black] (-1,-1) circle (2.5pt);
\draw [fill=black] (0,-1) circle (2.5pt);
\draw [fill=black] (1,-1) circle (2.5pt);
\draw [fill=black] (-0.5,-1) circle (2.5pt);
\draw [fill=black] (0.5,-1) circle (2.5pt);
\draw [fill=black] (-1,-2) circle (2.5pt);
\draw [fill=black] (0,-2) circle (2.5pt);
\draw [fill=black] (1,-2) circle (2.5pt);
\draw [fill=black] (-0.5,-2) circle (2.5pt);
\draw [fill=black] (0.5,-2) circle (2.5pt);

\draw [fill=black] (-1,-3) circle (2.5pt);
\draw [fill=black] (0,-3) circle (2.5pt);
\draw [fill=black] (1,-3) circle (2.5pt);
\draw [fill=black] (1.5,-3) circle (2.5pt);
\draw [fill=black] (-0.5,-3) circle (2.5pt);
\draw [fill=black] (0.5,-3) circle (2.5pt);

\end{scriptsize}
\end{tikzpicture}
\vspace{-0.05cm}
  \caption{
  $i=2$ and radius $\ell_2=3$. (Note, $\ell_2 = \ell_3$.)}
  \label{fig:treefig2}
\end{subfigure}
\hspace{5pt}
\begin{subfigure}{.31\textwidth}
  \centering
  \begin{tikzpicture}
\draw [line width=0.8pt] (0,0)-- (-1,-1);
\draw [line width=0.8pt] (0,0)-- (0,-1);
\draw [line width=0.8pt] (0,0)-- (1,-1);
\draw [line width=0.8pt] (0,0)-- (-0.5,-1);
\draw [line width=0.8pt] (0,0)-- (0.5,-1);

\draw [line width=0.8pt] (-1,-2)-- (-1,-1);
\draw [line width=0.8pt] (-0.5,-2)-- (-0.5,-1);
\draw [line width=0.8pt] (0,0)-- (0,-1);
\draw [line width=0.8pt] (0,0)-- (0,-1);
\draw [line width=0.8pt] (0,-1)-- (0,-2);
\draw [line width=0.8pt] (0,-1)-- (0.5,-2);
\draw [line width=0.8pt] (1,-1)-- (1,-2);
\draw [line width=0.8pt] (1,-1)-- (1,-2);
\draw [line width=0.8pt] (-1,-2)-- (-1,-3);
\draw [line width=0.8pt] (-1,-2)-- (-0.5,-3);
\draw [line width=0.8pt] (0,-2)-- (0,-3);
\draw [line width=0.8pt] (0,-2)-- (0.5,-3);
\draw [line width=0.8pt] (0,-2)-- (1,-3);
\draw [line width=0.8pt] (1,-2)-- (1.5,-3);
\draw [line width=0.8pt] (-1,-3)-- (-1,-4);
\draw [line width=0.8pt] (-1,-3)-- (-0.5,-4);
\draw [line width=0.8pt] (0,-3)-- (0,-4);
\draw [line width=0.8pt] (1.5,-3)-- (0.5,-4);
\draw [line width=0.8pt] (1.5,-3)-- (1,-4);
\draw [line width=0.8pt] (1.5,-3)-- (1.5,-4);

\begin{scriptsize}
\draw [fill=black] (0,0) circle (2.5pt);
\draw [fill=black] (-1,-1) circle (2.5pt);
\draw [fill=black] (0,-1) circle (2.5pt);
\draw [fill=black] (1,-1) circle (2.5pt);
\draw [fill=black] (-0.5,-1) circle (2.5pt);
\draw [fill=black] (0.5,-1) circle (2.5pt);
\draw [fill=black] (-1,-2) circle (2.5pt);
\draw [fill=black] (0,-2) circle (2.5pt);
\draw [fill=black] (1,-2) circle (2.5pt);
\draw [fill=black] (-0.5,-2) circle (2.5pt);
\draw [fill=black] (0.5,-2) circle (2.5pt);

\draw [fill=black] (-1,-3) circle (2.5pt);
\draw [fill=black] (0,-3) circle (2.5pt);
\draw [fill=black] (1,-3) circle (2.5pt);
\draw [fill=black] (1.5,-3) circle (2.5pt);
\draw [fill=black] (-0.5,-3) circle (2.5pt);
\draw [fill=black] (0.5,-3) circle (2.5pt);

\draw [fill=black] (-1,-4) circle (2.5pt);
\draw [fill=black] (0,-4) circle (2.5pt);
\draw [fill=black] (1,-4) circle (2.5pt);
\draw [fill=black] (1.5,-4) circle (2.5pt);
\draw [fill=black] (-0.5,-4) circle (2.5pt);
\draw [fill=black] (0.5,-4) circle (2.5pt);

\end{scriptsize}
\end{tikzpicture}
\vspace{-0.05cm}
  \caption{
  $i=4$ and radius $\ell_4=4$.}
  \label{fig:treefig3}
\end{subfigure}
\caption{For $t=4$, examples of trees that have radius at most $\ell_i$, for three different values of $i$.}
\label{fig:examplesTrees}
\end{figure}

Theorem~\ref{thm:infinity_from_timed} when applied to trees then gives that if $T$ has radius at most $\ell_i$, then $c_t^\infty(T) \leq i$.
When we apply Corollary~\ref{cor:retracts_give_CER} to trees and choose the optimal value of $i$, we obtain the following result.

\begin{corollary} \label{cor:treesupper}
For all $t\in \mathbb{N^*}$, if $T$ decomposes into subtrees such that, for each $i \in \mathbb{N}^{*}$, there are $\alpha_i$ subtrees in the decomposition that have radius at most $\ell_i$ and at least $\ell_{i-1}+1$, then
\begin{align} \label{eqn:decomp}
c_t^{\infty}(T) \leq \sum_i \alpha_i \cdot i.\end{align}
\end{corollary}

We note the following lemma.

\begin{lemma}
There exists a polynomial-time algorithm to determine the decomposition of a tree into retracts such that the right hand side of (\ref{eqn:decomp}) is minimized.
\end{lemma}
\begin{proof}
Let $T$ be the tree that we are trying to decompose.
We consider the family of subtrees $\mathcal{F} = \{N^r_T[v] : 1 \leq r \leq t \text{ and } v \in V(T)\}$, where $N^r_T[v]$ is the set of all vertices at distance at most $r$ from $v$ in $T$.
Notice that any two subtrees in this family $T_1,T_2$ have the property that if $T_1 \setminus T_2$ is disconnected, then $T_2 \setminus T_1$ is connected and similarly  if $T_2 \setminus T_1$ is disconnected, then $T_1 \setminus T_2$ is connected.
This follows from the definition of $\mathcal{F}$.
Hence, any path in $T$ that connects two vertices of $T_2 \setminus T_1$ is disjoint from any path in $T$ that connects two vertices of $T_1 \setminus T_2$. This property is known as being \emph{fork-free}.
A polynomial-time algorithm exists for optimally covering a tree with a given family of fork-free subtrees along with associated weights~\cite{BEW86}.

We give a subtree a weight of $i$ if it has  radius at most $\ell_i$ and at least $\ell_{i-1}+1$.
We apply this algorithm, which takes polynomial time.
This yields a cover of subtrees, and we want a decomposition of a tree into subtrees.
To construct a decomposition from the cover, we consider any two intersecting subtrees $T_1, T_2$ in the cover. If $T_1 \setminus T_2$ is connected, then we set $T'_1 =  T_1 \setminus T_2$ and $T'_2 = T_2$. Otherwise, $T_2 \setminus T_1$ is connected, then we set $T'_2 =  T_2 \setminus T_1$ and $T'_1 = T_1$. We replace subtrees $T_1, T_2$ in the cover with subtrees $T'_1, T'_2$. The weights of $T'_1$ and $T'_2$ are not greater than $T_1$ and $T_2$.
After all intersecting subtrees have been processed in this fashion, the resulting cover is also a decomposition.
Such processing takes linear time.
\end{proof}

We conjecture that the resulting upper bound obtained through Corollary~\ref{cor:treesupper} is close to being tight in the case that $T$ comes from an optimally chosen decomposition into subtrees.
The intuition behind this idea has two parts.

The first is that subtrees with small maximum degree and relatively large radius (close to $t$) can often be decomposed into a small number of subtrees with relatively small radius.
Say some subtree $P$ that occurs in a decomposition of a tree $T$ is a path of length $2\ell_i$, for any integer $i$ larger than $1$.
While $P$ has $c_t^\infty(P)=2$, the bound in Corollary~\ref{cor:treesupper} only yields  $c_t^\infty(P)\leq i$.
However, $P$ can be decomposed into two paths $P'$ and $P''$ of lengths $\ell_i$ and $\ell_i-1$, respectively, which have $c_t^\infty(P')=c_t^\infty(P'')=1$.
As such, an  optimally chosen decomposition would avoid this subtree with small maximum degree and relatively large radius.

The second is that the bound is tight for large families of subtrees with relatively small radius. Each of the subtrees $T'$ in the decomposition of Corollary~\ref{cor:treesupper} that have radius at most $\ell_i$ and at least $\ell_{i-1}+1$ are said to have the upper bound $c_t^\infty(T') \leq i$.
In the next two lemmas, we show that there are significant cases where this bound for $T'$ is tight.

\begin{lemma} \label{lem:simplerLowerBoundTrees}
For all $t,i, \ell \in \mathbb{N^*}$ with $\ell \geq (1 -2^{-i})t +  1-2^{-i}$, if $T$ is a rooted tree such that its root has degree $\beta >i$ and all leaves are distance $\ell$ from the root, then $c^{\infty}_t (T) >i$.
\end{lemma}

\begin{proof}
We show that for any cop strategy with $i$ cops, a robber will win within $i+1$ plays.
We first note that because we have assumed that the root of $G$ has degree $\beta$, removing this root leaves us with $\beta$ subtrees.
We call these $\beta$ subtrees the \emph{branches} of $T$. We will refer to the leaves of a branch, by which we mean the leaves of $T$ that are contained within the branch.

Note that on any play, there exists a branch that does not contain a cop, as there are more branches than cops since $\beta > i$.
On the first robber turn of play $j$, $1 \leq j \leq i$, the incoming robber places himself at a distance of $r_j=(2^{i-j+1}-1) (t-\ell) +(2^{i-j+1}-1)$ from the root, on a branch that does not contain a cop.
We note that this is possible since $r_1\leq (1 -2^{-i})t +  1-2^{-i} \leq \ell$  (with the left inequality holding since  the initial assumption on $\ell$ yields $t-\ell \leq t- (1 -2^{-i})t -  1+2^{-i}$)
and the sequence $(r_j)$ is weakly decreasing.
The distance from the robber to his closest leaf is $\ell - r_j$.
As such, if the robber is not captured within $t - (\ell - r_j)$ time-steps, then the robber moves unimpeded towards a leaf of his branch, finishing at the leaf on time-step $t$, hence, winning against the cops.
Otherwise, the robber stays at his vertex on each robber move.
On play $i+1$, the robber places himself at a leaf of an unoccupied branch, and does not move for each subsequent turn.

We now detail how any cop strategy will fail to capture the robber in at most $t$ time-steps on (or before) the $(i+1)^{th}$ play while the robber follows the above strategy.
At the end of the first play, if the robber has been caught, then one cop (say $C_1$) must have captured the robber, and so must be at a distance $r_1$ from the root.
During the second play, the robber must be captured by a cop (say cop $C_2$) within the first $t - ( \ell - r_2)$ time-steps, or else the robber wins.
This implies that at the end of the second play, cop $C_1$ will be at least distance $ r_1 -(t - \ell + r_2) = (2^{i-1}-1) (t-\ell) + 2^{i-1} = r_2+1$ from the root, meaning that cop $C_1$ is further from the root than $C_2$ at the end of the play.

We can continue this procedure for all plays.
During the $j^{th}$ play, the robber must be captured by cop $C_j$ within the first $t - ( \ell  -r_j)$ time-steps, or else the robber wins.
Hence, at the end of the play, all preceding cops will be at least distance $ r_{j-1} -(t - \ell + r_j)  = (2^{i-j+1}-1) (t-\ell) + 2^{i-j+1} = r_j+1$ from the root, meaning that all preceding cops are further from the root than $C_j$.

At the end of play $i$, cop $C_i$ is at a distance $r_i=t- \ell  +1$ from the root, and all other cops are at least as far from the root as $C_i$.
When the next robber appears on a leaf of an unoccupied branch, the robber is at a distance of at least $ \ell  + (t- \ell +1) = t+1$ from any cop. Thus, the robber wins.
\end{proof}

We have the following corollary.

\begin{corollary} \label{cor:lowerbound}
For all $t,i,\ell \in \mathbb{N^*}$ with $\ell \geq (1 -2^{-i})t +  1-2^{-i}$, if a tree $T$ contains a set $S$ of $i+1$ vertices such that the pairwise distance between vertices in $S$ is at least $2 \ell$, then $c^{\infty}_t (T) >i$.
\end{corollary}

There is some $\epsilon_i$ with $0 \leq \epsilon_i \leq 1.5$ such that the bounding value in Corollary~\ref{cor:lowerbound} is $(1 -2^{-i})t +  1-2^{-i} = \ell_i+\epsilon_i$. In particular, for a tree $T$ that has $i$ vertices of pairwise distance at least $2 \ell$ and a radius $\ell$ satisfying $\ell_{i-1}+\epsilon_{i-1} \leq \ell$, this lemma implies that $c_t^\infty(T)>i-1$.
Similarly, for those of radius $\ell$ satisfying $\ell \leq \ell_i$, Corollary \ref{cor:treesupper} implies that $c_t^\infty(T)\leq i$.
Combining these two facts we have the following.

\begin{corollary}
For a tree $T$ that has $i$ vertices of pairwise distance at least $2 \ell$ and a radius of $\ell$ satisfying $\ell_{i-1}+\epsilon_{i-1} \leq \ell \leq \ell_i$, we have $c_t^\infty(T)=i$.
\end{corollary}

When a subtree with radius $\ell$ satisfying  $\ell_{i-1}+\epsilon_{i-1} \leq \ell \leq  \ell_i$ appears in a decomposition but does not contain $i$ vertices of pairwise distance at least $2 \ell$, it may be possible to provide a better decomposition of the tree for Corollary~\ref{cor:treesupper}.
As a result, the upper bound given in Corollary~\ref{cor:treesupper} may be close to the best possible value.

\subsection{Products}\label{sec:grids}

In this subsection, we study the game of Cops and Eternal Robbers in the strong product and Cartesian product of graphs.

\subsubsection{Strong Products}
The strong product of two graphs $G_1$ and $G_2$, denoted by $G=G_1  \boxtimes G_2$, has vertex set $V(G) = V(G_1)\times V(G_2)$ and vertices $(u,u')$ and $(v,v')$ are adjacent if:
\begin{enumerate}
\item $u = v$ and $u'$ is adjacent to $v'$; or
\item $u' = v'$ and $u$ is adjacent to $v$; or
\item $u$ is adjacent to $v$ and $u'$ is adjacent to $v'$.
\end{enumerate}

It follows from this definition that if $H_i$ is a retract of $G_i$ with retraction mapping $f_i$ for $i \in [n]$, then $H_1 \boxtimes \cdots \boxtimes H_n$ is a retract of $G_1 \boxtimes \cdots \boxtimes G_n$ with retraction mapping $f_1 \times \cdots \times f_n$, the Cartesian product of the functions.
\begin{lemma}  \label{lm:strong}
Let $i,p,t\in \mathbb{N^*}$ such that $2 \leq i \leq p$. For any set of graphs $H_i$ with $c_t^\infty(H_1) = k$ and $c_t^\infty(H_i) = 1$,
\[c_t^\infty(H_1 \boxtimes \ldots \boxtimes H_p) = k.\]
\end{lemma}

\begin{proof}
We prove this in the case that $p=2$, and the general case follows by induction.
We begin by showing that $c_t^\infty(H_1 \boxtimes  H_2) \leq k$.
For the game of Cops and Eternal Robbers on $H_1 \boxtimes H_2$, there are two \emph{shadow games} played on $H_1$ and $H_2$, defined by the first and second projection map, respectively.
We play with $k$ cops on $H_1 \boxtimes H_2$, and hence, on both shadow games as well.
We know that this is enough cops to have a winning strategy on both shadow games.
A cop at position $(x_1,x_2)$ on $H_1 \boxtimes H_2$ observes what move the corresponding cops on the shadow games would take using their winning strategy. If the corresponding cop of $H_1$ would move from $x_1$ to $y_1$, and  the corresponding cop of $H_2$ would move from $x_2$ to $y_2$, then the cop moves from $(x_1,x_2)$ to $(y_1,y_2)$.
Within $t$ time-steps, there must be some cop such that the corresponding cop on $H_1$ has captured the robber's shadow, as we know that $k$ cops are sufficient for the robber to be caught in $H_1$ in time $t$. This cop's corresponding shadow on $H_2$ was also able to capture the robber's shadow in $H_2$, as any cop on $H_2$ is able to capture the robber in time $t$. This means that the cop has captured the robbers on both shadow games, and so has captured the robber on $H_1 \boxtimes H_2$.

To show that $c_t^\infty(H_1 \boxtimes H_2) \geq k$, suppose
a game on $H_1 \boxtimes H_2$ is played with $k-1$ cops.
The robber plays such that his shadow on $H_1$ avoids capture, which he can do since $c_t^\infty(H_1)=k$. Since his shadow avoids capture, he avoids capture on $H_1 \boxtimes H_2$ as well.
\end{proof}

As a particular example of interest, we consider $ P_{n_1} \boxtimes \ldots \boxtimes P_{n_p}$, which can be seen as a $p$-dimensional generalization of a king's graph. Each such graph with sufficiently large $n_i$ contains a subgraph of the form $P_{t+1} \boxtimes \ldots \boxtimes P_{t+1}$.
Lemma~\ref{lm:strong} shows this subgraph has the property that it requires just one cop to capture the robber in every play.
In fact it is the largest subgraph with this property, as any subgraph $S$ of $ P_{n_1} \boxtimes \ldots \boxtimes P_{n_p}$ with $c_t^\infty(P)=1$  must be a subgraph of $ P_{t+1} \boxtimes \ldots \boxtimes P_{t+1}$.

\begin{corollary}
For all $t\in \mathbb{N^*}$, we have that \[c_t^\infty( P_{t+1} \boxtimes \ldots \boxtimes P_{t+1}) =1. \]
\end{corollary}

Since each $P_{n_i}$ decomposes into $\left\lceil \frac{n_i}{t+1} \right\rceil$ copies of $P_{t+1}$ for each $i \in [n]$, we can decompose the entire $p$-dimensional generalization of a king's graph into retracts. We then have the following, which follows from Lemma~\ref{lem:disjointRetracts}.

\begin{lemma} \label{lem:strongGrid}
For all $n_1,n_2,\ldots, n_p,t\in \mathbb{N^*}$, we have that \[c_t^\infty( P_{n_1} \boxtimes \ldots \boxtimes P_{n_p}) \leq \prod_{i=1}^p \left\lceil \frac{n_i}{t+1} \right\rceil. \]
\end{lemma}

It is not obvious if the bound in Lemma~\ref{lem:strongGrid} is tight or not, although we conjecture it is. However, we derive the asymptotic order in the following.

\begin{theorem}
For all $n_1,n_2,\ldots, n_p,t\in \mathbb{N^*}$, we have that
\[ c_t^\infty( P_{n_1} \boxtimes \ldots \boxtimes P_{n_p}) = \Theta\left(\prod_{i=1}^p \left(\frac{n_i}{t}\right)\right),\]
or more precisely,
\[ \prod_{i=1}^p \left(\frac{n_i}{2t+1}\right) \leq c_t^\infty( P_{n_1} \boxtimes \ldots \boxtimes P_{n_p}) \leq  \prod_{i=1}^p \left\lceil \frac{n_i}{t+1} \right\rceil. \]
\end{theorem}

\begin{proof}
We will show that $c_t^\infty( P_{n_1} \boxtimes \ldots \boxtimes P_{n_p}) \geq \prod_{i=1}^p \frac{n_i}{2t+1}$. Along with Lemma \ref{lem:strongGrid}, this will complete the proof.
The graph contains $\prod_{i=1}^p n_i$ vertices.
Every vertex must be reachable by at least one cop within $t$ time-steps.
Any cop can reach at most $(2t+1)^p$ unique vertices in the graph.
We therefore require at least $(\prod_{i=1}^p n_i) / ((2t+1)^p)$ cops.
This completes the proof.
\end{proof}

\subsubsection{Cartesian Products}

In this subsection, we will give upper and lower bounds of the same asymptotic order on the number of cops needed for the Cartesian product of two paths. The Cartesian product of two graphs $G_1$ and $G_2$, denoted by $G=G_1 \square G_2$,
has vertex set $V(G) = V(G_1) \times V(G_2)$ and vertices $(u,u')$ and $(v,v')$ are adjacent if:
\begin{enumerate}
\item $u = v$ and $u'$ is adjacent to $v'$; or
\item $u' = v'$ and $u$ is adjacent to $v$.
\end{enumerate}
We define an $m \times n$ (Cartesian) \emph{grid} to be $P_m \square P_n$, where we assume $m,n \geq 2$. We may view vertices of the grids via Cartesian coordinates $\{ (i,j): 1 \le i \le m, 1 \le j \le n \}.$ Note that it is known that $c(P_m \square P_n)=2$; see \cite{BN}.

We derive an upper bound on the eternal cop number of grids by decomposing them into small subgrids, which are retracts.

\begin{lemma}\label{lem:cartsubgrid}
For all $t\in \mathbb{N^*}$ and all integers $i\geq 1$, we have that $c_t^\infty(P_{\ell_i+1} \square P_{\ell_i+1}) \leq 2i$.
\end{lemma}

\begin{proof}
In~\cite{M11}, it was proven that $\mathrm{capt}(P_m \square P_n)=\left \lfloor \frac{m+n}{2} \right \rfloor$. Thus, $c_{\ell_i}(P_{\ell_i+1} \square P_{\ell_i+1})=2$, since $\mathrm{capt}(P_{\ell_i+1} \square P_{\ell_i+1})=\ell_i$. We therefore have that $P_{\ell_i+1} \square P_{\ell_i+1}\in \mathcal{G}_i^2$ and by Theorem~\ref{thm:infinity_from_timed}, the proof follows.
\end{proof}

We now prove that any subgrid of a grid is indeed a retract.

\begin{lemma}\label{lem:gridretract}
For all $a,b,n,m\in \mathbb{N^*}$ such that $m,n>2$, $a\leq m$, and $b\leq n$, we have that $P_a \square P_b$ is a retract of $P_m \square P_n$.
\end{lemma}

\begin{proof}
We first prove that $H=P_{m-1} \square P_n$ is a retract of $G=P_m \square P_n$. The retraction $f:G\rightarrow H$ is defined as follows:

\begin{enumerate}
\item if $(x,y) \in V(H)$, then $f((x,y))=(x,y)$;
\item if $y<n$, then $f((m,y))=(m-1,y+1)$;
\item and $f((m,n))=(m-1,n-1)$.
\end{enumerate}

Note that, symmetrically, this also proves that $H=P_m \square P_{n-1}$ is a retract of $G=P_m \square P_n$. Observe that if $f:G_1\rightarrow G_2$ and $g:G_2\rightarrow G_3$ are retractions, then there exists a retraction $h:G_1\rightarrow G_3$. The proof then follows from applying successive retractions.
\end{proof}

We can now give bounds on the eternal cop number of grids, where both the lower and upper bounds are of the same asymptotic order. In what follows, we consider that both $m$ and $n$ are much larger than $t$.

\begin{theorem}\label{gridd}
For all $m,n,t\in \mathbb{N^*}$ such that $m,n\geq 2$, $t^2=o(m)$, and $t^2=o(n)$, we have that
\[ c_t^\infty (P_m \square P_n)=\Theta\left(\frac{mn}{t^2}\right), \]
\text{or more precisely,}
\[ \frac{mn}{2t^2+2t+1} \leq c_t^\infty (P_m \square P_n) \leq \frac{64mn}{9t^2+12t+4}.\]
\end{theorem}

\begin{proof}
We first prove the lower bound. There are at most $1+4+8+\ldots+4t=2t(t+1)+1$ vertices at distance $t$ or less from any cop and $mn$ vertices in total. Therefore, at least $\frac{mn}{2t^2+2t+1}$ cops are needed in order to distance $t$-dominate $P_m \square P_n$.

For the upper bound, we partition $P_m \square P_n$ into subgrids of size at most $P_{\ell_2+1} \square P_{\ell_2+1}$ (the value of $i=2$ gives the best upper bound). Observe that $$|P_{\ell_2+1} \square P_{\ell_2+1}|=(\ell_2+1)^2\geq \frac{9t^2+12t+4}{16},$$ and so $P_m \square P_n$ can be partitioned into $\frac{16mn}{9t^2+12t+4}$ subgrids of size at most $P_{\ell_2+1} \square P_{\ell_2+1}$. By Lemma~\ref{lem:cartsubgrid}, we have that $c_t^\infty(P_{\ell_2+1} \square P_{\ell_2+1}) \leq 4$. Since the capture time of any grid is monotonically decreasing in terms of its order, we have that $c_t^\infty(P_a \square P_b) \leq 4$ for all $a,b\leq \ell_2+1$. By Lemma~\ref{lem:gridretract}, we have that $P_a \square P_b$ is a retract of $P_m \square P_n$. Therefore, by Lemma~\ref{lem:disjointRetracts}, we have that $c_t^\infty (P_m \square P_n) \leq \frac{64mn}{9t^2+12t+4}$.
\end{proof}

We note that by an omitted argument, we may give a slight improvement on the multiplicative constants in Theorem~\ref{gridd} by partitioning the grid into $k$-neighbor sets rather than subgrids.

\section{Further Directions}\label{sec:conclusion}
We showed that the game is {\bf NP}-hard when $t$ is a fixed constant and {\bf EXPTIME}-complete when $t$ is sufficiently large. The exact complexity class of the game for all values of $t$ remains open. Is computing the eternal cop number {\bf EXPTIME}-complete for all $t?$ We note that solving this problem would solve the open problem of the exact complexity class of the Eternal Domination game (the case $t=1$), which is known to be {\bf NP}-hard~\cite{BDEMY17}. Our results strengthen the intuition that the Eternal Domination game is \textbf{EXPTIME}-complete.

There is more work to be done on Cops and Eternal Robbers for trees. In Corollary~\ref{cor:treesupper}, we gave an upper bound for the eternal cop number of trees that we conjecture to be tight.  We mention in this context the {\it Spy game}~\cite{CMM^+17}, where the attacker in this game moves like an agent on the graph much like in Cops and Eternal Robbers. In the Spy game, the guards must maintain that there is at least one guard at distance at most $d\geq 0$ from the spy (attacker) at the end of the guards' turn, and the spy moves at speed $s\geq 2$. In particular, the spy can only attack or move to a vertex at distance at most $s$ from its current position. In~\cite{CMM^+17}, the game was shown to be {\bf NP}-hard, paths were resolved, and almost tight bounds were given for cycles. In~\cite{CMNP18}, a polynomial-time algorithm for trees was obtained using Linear Programming and a fractional relaxation of the game, and some asymptotic bounds were obtained for grids.

For the Spy Game on trees, a fractional variant and Linear Program gave a polynomial-time algorithm. In that case, there exists an optimal strategy in trees in which there is a unique configuration for the guards for each position of the spy~\cite{CMNP18}. Cops and Eternal Robbers is different, however, since the robber chooses his position at the beginning of each play. Thus, the question of whether the complexity of the eternal cop number on trees is polynomial remains open.

\end{document}